\documentclass{article}

\usepackage[a4paper]{geometry}
\usepackage{graphicx}
\usepackage[textwidth=8em,textsize=small]{todonotes}
\usepackage{amsmath}
\usepackage{amsthm}
\usepackage{amsfonts}
\usepackage{subcaption}
\usepackage{natbib}
\usepackage{url}
\usepackage[hidelinks]{hyperref}
\usepackage{bm}

 \newtheorem{prop}{Proposition} 
\newtheorem{theorem}{Theorem}
\newtheorem{corollary}{Corollary}
  
\title{Linked Matrix Factorization}

\author{Michael J. O'Connell and Eric F. Lock  \\
\\
	   Division of Biostatistics, University of Minnesota\\  Minneapolis, MN 55455, U.S.A \\
	   }

\date{}

\begin{document}
\maketitle

\begin{abstract}

 In recent years, a number of methods have been developed for the dimension reduction and decomposition of multiple linked high-content data matrices.  Typically these methods assume that just one dimension, rows or columns, is shared among the data sources.  This shared dimension may represent common features that are measured for different sample sets (i.e., \emph{horizontal integration}) or a common set of samples with measurements for different feature sets (i.e., \emph{vertical integration}).  In this article we introduce an approach for simultaneous horizontal and vertical integration, termed Linked Matrix Factorization (LMF), for the more general situation where some matrices share rows (e.g., features) and some share columns (e.g., samples).  Our motivating application is a cytotoxicity study with accompanying genomic and molecular chemical attribute data. In this data set, the toxicity matrix (cell lines $\times$ chemicals) shares its sample set  with a genotype matrix (cell lines $\times$ SNPs), and shares its feature set with a chemical molecular attribute matrix (chemicals $\times$ attributes). LMF gives a unified low-rank factorization of these three matrices, which allows for the decomposition of systematic variation that is shared among the three matrices and systematic variation that is specific to each matrix.  This may be used for efficient dimension reduction, exploratory visualization, and the imputation of missing data even when entire rows or columns are missing from a constituent data matrix.  We present theoretical results concerning the uniqueness, identifiability, and minimal parametrization of LMF, and evaluate it with extensive simulation studies.

\end{abstract}

\section{Introduction}
\label{intro}
Recent technological advances in biomedical research have led to a growing number of platforms for collecting large amounts of health data.  Molecular profiling modalities such as genetic sequencing and gene expression microarrays, and imaging modalities such as MRI scans, yield high-dimensional data with complex structure.    Methods that simplify such data by identifying latent patterns that explain most of the variability are very useful for exploratory visualization of systematic variation, dimension reduction, missing data imputation, and other tasks.  For a single data matrix $X: m \times n$,  this simplification can be accomplished via a principal components analysis (PCA) or via other approaches to low-rank matrix factorization (for an overview see \citet{wall2003singular}).  For example, $X$ may represent a microarray with expression measurements for $m$ genes for $n$ biological samples.  It is also increasingly common to have multiple linked high-dimensional data matrices for a single study, e.g.,  \begin{align}X_1: m_1 \times n, X_2: m_2 \times n, \hdots, X_k: m_k \times n \label{eqVertInt} \end{align}
with $n$ shared columns.  In the \emph{multi-source} context (\ref{eqVertInt}) $X_1$ may represent expression for $m_1$ genes, $X_2$ may represent the abundance of $m_2$ proteins, and $X_3$ may represent the abundance of $m_3$ metabolites, for a common set of $n$ biological samples.  For multi-source data a straightforward ad-hoc approach is to perform a separate PCA of each matrix $X_i$ (for example, see \citet{zhao2014combining}).  However,  patterns of systematic variability may be shared between blocks; for example, it is reasonable to expect that some sample patterns that are present in gene expression data are also present in proteins.  Thus, separate factorizations can be inefficient and underpowered to accurately recover these joint signals, and they also provide no insight into the connections between data matrices that are often of scientific interest.    An alternative approach is to perform a single joint PCA analysis of the concatenated data $X: (m_1+m_2+\cdots+m_k) \times n$, and this approach as been referred to as consensus PCA \citep{Wold,Westerhuis}.   However, a consensus PCA approach assumes that all systematic patterns are shared across data matrices, and lacks power to accurately recover signals that may exist in only one data matrix.  
 
The recent ubiquity of high-dimensional multi-source data has motivated more flexible methods for scenario (\ref{eqVertInt}).  A guiding principle for several such methods is to simultaneously model features that are shared across multiple sources (i.e., \emph{joint} and features that are specific to a particular source (i.e., \emph{individual}).  Methods that follow this strategy have been developed that extend well-established exploratory techniques such as partial least squares \citep{Lofstedt2011}, canonical correlation analysis \citep{zhou2015}, non-parametric Bayesian modeling \citep{ray2014}, non-negative factorization \citep{yang2016}, and simultaneous components analysis \citep{schouteden2014}.  The Joint and Individual Variation Explained (JIVE) method \citep{lock2013joint,o2016r} is a direct extension of PCA, distinguishing components that explain covariation (joint structure) between the data sources and principal components that explain variation that is individual to each data source.  This distinction simplifies interpretation, and also improves accuracy to recover underlying signals because structured individual variation can interfere with finding important joint signal, just as joint structure can obscure important signal that is individual to a data source.
  
Multi-source data integration (\ref{eqVertInt}) has been termed \emph{vertical integration} \citep{tseng2015integrating}. Related dimension reduction and pattern recognition methods have also been developed specifically for the horizontal integration of a single data source  measured for multiple sample groups \citep{kim2017metaPCA,huo2016meta}: \begin{align} X_1: m \times n_1, X_2: m \times n_2, \hdots, X_k: m \times n_k. \label{eqHorizInt} \end{align}   
 Other integrative approaches have been developed for a collection of matrices that share both dimensions: \begin{align} X_1: m \times n, X_2: m \times n, \hdots, X_k: m \times n. \label{eqPVD} \end{align}  Population value decomposition (PVD) \citep{crainiceanu2011population} was designed for the analysis of aligned image populations and produces a joint low-rank factorization for scenario (\ref{eqPVD}) with shared row and column components; similar techniques have also been developed in the computer science literature \citep{ding20052,ye2005generalized}. Another approach for scenario(\ref{eqPVD}) is to treat the data as a single multi-way array (i.e., tensor) $\mathbb{X}: m \times n \times k$ and apply well-established tensor factorizations such as the CANDECOMP/PARAFAC \citep{harshman1970foundations} and Tucker \citep{tucker1966some} factorization which extend PCA and related matrix dimension reduction methods to higher-order arrays.  Neither PVD nor a tensor factorization approach distinguishes joint and individual structure among the constituent data matrices. The Linked Tucker2 Decomposition \citep{yokota2014linked} and Bayesian Multi-view tensor factorization \citep{khan2014bayesian} methods do allow for the decomposition of joint and individual structure, under an extended scenario for (\ref{eqPVD}) where the collection of $m \times n$ matrices can be grouped into sets of size $d_i$: 
 \begin{align} \mathbb{X}_1: m \times n \times d_1, \mathbb{X}_2: m \times n \times d_2, \hdots, \mathbb{X}_k: m \times n \times d_k.\end{align}
 \citet{acar2011all} describe a method for the joint factorization of a matrix and a tensor, but their approach does not allow for the decomposition of joint and individual structure.  

In this article we address the simultaneous low-rank factorization and decomposition of joint and individual structure for the novel context of linked data in which the shared dimensions (e.g., rows or columns) are not consistent across the constituent data matrices.  Our motivating example is a large-scale cytotoxicity study \citep{abdo2015population} that consists of three interlinked high-content data matrices:   
\begin{enumerate}
\item $X: m_1 \times n_1$: A cytotoxicity matrix with a measure of cell death for $n_1$ chemicals across a panel of $m_1$ genetically distinct cell lines,
\item $Y: m_2 \times n_1$:  A chemical attribute matrix with $m_2$ molecular attributes measured for the $n_1$ chemicals, and 
\item $Z: m_1 \times n_2$: A genomic matrix with $n_2$ single nucleotide polymorphisms (SNP) measured for each of the $m_1$ cell lines.
\end{enumerate}
Note that $X$ shares its row set with $Z$ and its column set with $Y$, as illustrated in Figure~\ref{fig:xyz}.  These data were made public as part of a DREAM challenge for open science \citep{eduati2015prediction}.  We are particularly interested in investigating the interaction between chemical toxicity, genomics, and measurable chemical attributes, i.e.,  what systematic variability in $X$ is shared by $Y$ and $Z$?  

We introduce a method called Linked Matrix Factorization (LMF), that gives a unified and parsimonious low-rank factorization of these three data matrices.  We also extend the framework of the JIVE method to allow for the decomposition of joint and individual structure in this context with LMF-JIVE. This extension requires new approaches to estimation and new theoretical results concerning the uniqueness, identifiability, and minimal parametrization of the decomposition.  We illustrate how the results can facilitate the visual exploration of joint and individual systematic variation. We also describe how the factorization can be used in conjunction with an Expectation-Maximization (EM) algorithm \citep{dempster1977maximum} for the principled imputation of missing values and complete analysis of multi-source data, even when entire rows or columns are missing from the constituent data matrices.  

In what follows we first describe a novel joint low-rank factorization of these data in Section~\ref{sec2}, before describing the extension to joint and individual structure in Section~\ref{sec3}.  In Section~\ref{sec4} we discuss missing data imputation, and in Section~\ref{sec5} we discuss different approaches to select the joint and individual ranks (i.e., number of components) in the factorization.  Sections~\ref{sec2},~\ref{sec3}, ~\ref{sec4}, and \ref{sec5} each include simulation studies, to assess each component of the proposed methodology.  In Section 6 we describe the results of the cytotoxicity application, and in Section 7 we give some concluding remarks.

\begin{figure}[h]
   \centering
   \includegraphics[scale=.5]{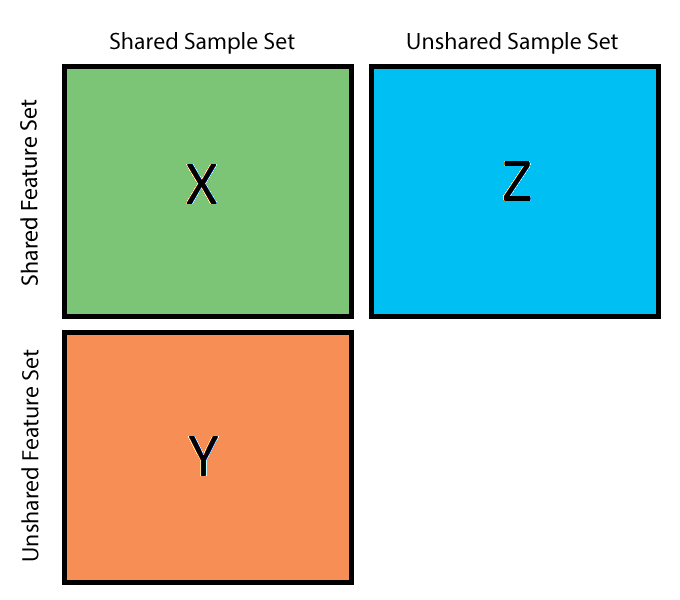}
   \caption{The structure of data for which the LMF algorithm was designed to analyze.  The X and Y matrices share a sample set and have a common row space.  Similarly, the X and Z matrices share a feature set and have a common column space.}
   \label{fig:xyz}
\end{figure}

\section{Joint Factorization}
\label{sec2}
\subsection{Model} \label{mod1}

We will refer to the three matrices involved in the LMF as $X$, $Y$, and $Z$, where $X$ shares its row space with $Y$ and its column space with $Z$ (Figure \ref{fig:xyz}).  Let the dimensions of X be $m_1\times n_1$, the dimensions of Y be $m_2\times n_1$, and the dimensions of Z be $m_1 \times n_2$.  Our task is to leverage shared structure across $X$, $Y$, and $Z$ in a simultaneous low-rank factorization. We define a joint rank $r$ approximation for the three data matrices as follows: 
\begin{align*}
X &= US_xV^T + E_x \\ 
Y &= U_yS_yV^T + E_y \\
Z &= US_zV_z^T + E_z 
\end{align*}
where 
\begin{itemize}
\item $U$ is an $m_1 \times r$ matrix representing the row structure shared between $X$ and $Z$
\item $V$ is an $n_1 \times r$ matrix representing the column structure shared between $X$ and $Y$
\item $U_y$ is an $m_2 \times r$ matrix representing how the shared column structure is scaled and weighted over $Y$
\item $V_z$ is an $n_2 \times r$ matrix representing how the shared row structure is scaled and weighted over $Z$
\item $S_x$, $S_y$, and $S_z$ are $r\times r$ scaling matrices for $X$, $Y$, and $Z$, respectively, and
\item $E_x$, $E_y$ and $E_z$ are error matrices in which the entries are independent and have mean $0$.
\end{itemize}

For the remainder of this article we will subsume the scaling matrices $ S_y $ and $ S_z $ into the $ U_y $ and $ V_z $ matrices, respectively, for a more efficient parameterization: $Y \approx U_y V^T$ and $Z \approx U V_Z^T$.  Then, for identifiability of the components it suffices to assume that the columns of $U$ and $V$ are orthonormal and $S_x$ is diagonal (see Section ~\ref{diagS}).  For notational convenience we also denote the following  concatenated matrices that span shared dimensions : $ \tilde{U} = [US_x \,\,\, U_y], \tilde{V} = [VS_x \,\,\, V_z], \tilde{Y} = [X^T \,\,\, Y^T]^T, \tilde{Z} = [X \,\,\, Z]. $

\subsection{Estimation} \label{estj}

To estimate the joint structure, we iteratively minimize the sum of squared residuals 
\[\mbox{SSE} = ||E_x||_F^2+||E_y||_F^2 +||E_z||_F^2,\]
using an alternating least squares algorithm, where $||\cdot||_F$ defines the Frobenius norm.  Given initial values, the algorithm proceeds by iteratively updating the components $U$, $V$, $U_y$, $V_y$, and $S_x$. In practice we first center and scale the three data sets, which prevents any of the matrices from having a disproportionately large influence on the joint components.  Next, we initialize $\tilde{V}$ as the first $r$ right singular vectors of the singular value decomposition (SVD) of $\tilde{Z}$.  We initialize $S_x$ as the identity matrix, so that $V$ is the first $n_1$ columns of $\tilde{V}$.  We then repeatedly cycle through the following local least-squares minimization steps with other components held fixed: 
\begin{enumerate}
\item Update $U_y$ via ordinary least squares: $ U_y = (V^T V)^{-1} V^TY $
\item Update $U$ via ordinary least squares:  $ U = (\tilde{V}^T \tilde{V})^{-1} \tilde{V}^T \tilde{Z} $
\item Scale $U$ by dividing each column by its Frobenius norm
\item Update $\tilde{U}$: $ \tilde{U} = [US_x \,\,\, U_y] $
\item Update $V$ via ordinary least squares:  $ V = (\tilde{U}^T \tilde{U})^{-1} \tilde{U}^T \tilde{Y} $
\item Update $V_z$ via ordinary least squares:  $ V_z = (U^T U)^{-1} U^T Z $
\item Scale V by dividing each column by its Frobenius norm
\item Update $S_x$ via least squares; define $W: np \times r$ such that the $i$'th column of $W$ is the vectorization of the product of the $i$th columns of $U$ and $V$, $W[,i] = \mbox{vec}(U[,i] V[,i]^T)$, then the diagonal entries of $S_x$ are $(W^T W)^{-1}W^T \mbox{vec}(X)$
\item Update $\tilde{U}$ and $\tilde{V}$ to incorporate the new $S_x$.
\end{enumerate}
The algorithm will improve the SSE at each step until convergence, resulting in the following rank $r$ estimates for the joint structure:
\begin{align}
\label{facEq}
\begin{split}
 J_x &= US_xV^T \\
 J_y &= U_yV^T  \\
 J_z &= UV_z^T.
 \end{split}
\end{align}

\subsection{Diagonalizing S$_x$} \label{diagS}

 Importantly, by constraining $S_x$ to be a diagonal matrix we do not limit our solution space, as shown in Proposition~\ref{thm:DiagSx}. This simplifies the parameterization, facilities the identifiability of the components, and improves computation time by only estimating the diagonal elements of $S_x$. This property does not extend to other scenarios where constituent data matrices share both rows and columns (\ref{eqPVD}), such as the PVD factorization \citep{crainiceanu2011population}.  
     
\begin{prop}
\label{thm:DiagSx}
Assume we have a decomposition in the form of Equation~(\ref{facEq}).  Then the decomposition can be rewritten such that $S_x$ is a diagonal matrix.
\end{prop}
  
\begin{proof}
Let $S^*_x$ be a matrix of rank r.  Suppose a set of matrices has the decomposition $ U^*S^*_xV^{*T} $, $ U_y^*V^{*T} $, and $ U^*V_z^{*T} $.  Then we can take a rank $r$ SVD of $ S^*_x $.  For this SVD, let P = the left singular vectors, Q = the right singular vectors, and $S_x$ = the singular values.  If we let $ U = U^*P $ and $ V = V^*Q $, then, we can rewrite this matrix decomposition as $ US_xV^T $, where $S_x$ is a diagonal matrix because it represents the singular values of $S^*_x$.  Then we can set $U_y = U_y^*P$ and $V_z = V_z^*Q$.   This gives us a decomposition in the form of Equation~\ref{facEq} in which $S_x$ is diagonal.  
\end{proof}

\subsection{Simulation Study} \label{simst1}

We ran a simulation to test the ability of the LMF algorithm to recover the true underlying joint structure of the data.  To generate the data, we simulated from the model, generating random matrices for the joint components.  There were 100 simulated data sets.  All elements of the joint structure components, $U$, $S_x$, $V$, $U_y$, and $V_z$, were simulated from a Normal$(0,1)$ distribution.  Error matrices were simulated from a Normal$(0,1)$ distribution.  The simulated data sets had 50 rows and 50 columns, and the rank of the joint structure was 2.  The LMF algorithm was run with a convergence threshold of .00001 and a maximum 5000 iterations.

We evaluated the performance of the LMF algorithm in these simulations by two criteria.  First, we calculated the relative reconstruction error, which measures the algorithm's ability to retrieve the true joint structure of the data. We scale the reconstruction residuals by the total structure to get the relative reconstruction error:   
$$ E_{rec} = \frac{\| J_x - J_x^{true} \| ^2_F + \| J_y - J_y^{true} \| ^2_F + \| J_z - J_z^{true} \| ^2_F}{\| J_x^{true} \| ^2_F + \| J_y^{true} \| ^2_F + \| J_z^{true} \| ^2_F} $$
where
$$ J_x^{true} = US_xV^T, J_y^{true} = U_yV^T, \text{ and } J_z^{true} = UV_z^T. $$
For the second criterion we consider the relative residual error, which measures the amount of variability in the data that is captured by the estimated joint structure:
$$ E_{res} = \frac{\| J_x - X \| ^2_F + \| J_y - Y \| ^2_F + \| J_z - Z \| ^2_F}{\| X \| ^2_F + \| Y \| ^2_F + \| Z \| ^2_F}. $$

Over the $100$ simulations, the mean reconstruction error was 0.122 with a standard deviation of 0.042.   The mean residual error was 0.588 with a standard deviation of 0.074.   This relatively small value for the mean reconstruction error in contrast to the mean residual error suggests that the LMF algorithm does a good job of recovering the underlying joint structure even in the presence of random noise.

We also ran a simulation with varying error variance to see how differences in error variance affect the accuracy of the decomposition.  This simulation was similar to the above simulation, but instead of the error matrices being drawn from a normal(0,1) distribution, they were drawn from a normal(0, $\sigma^2$) distribution, where $\sigma^2$ took on a value of $\frac{i}{100}$ for the $i^{th}$ simulated data set, varying between 0.01 and 1.  The relationship between error variance ($\sigma^2$) and reconstruction error is shown in Figure \ref{fig:evarSim}a, and the relationship between error variance and residual error is shown in Figure \ref{fig:evarSim}b.  These results demonstrate that in the presence of minimal noise the data are exactly captured with the converged low rank approximation; with increasing noise both the reconstruction error and residual error increase linearly, but the reconstruction error remains small. 

\begin{figure}[h]
   \centering
   \begin{subfigure}[b]{0.49\textwidth}
   \includegraphics[scale=0.4]{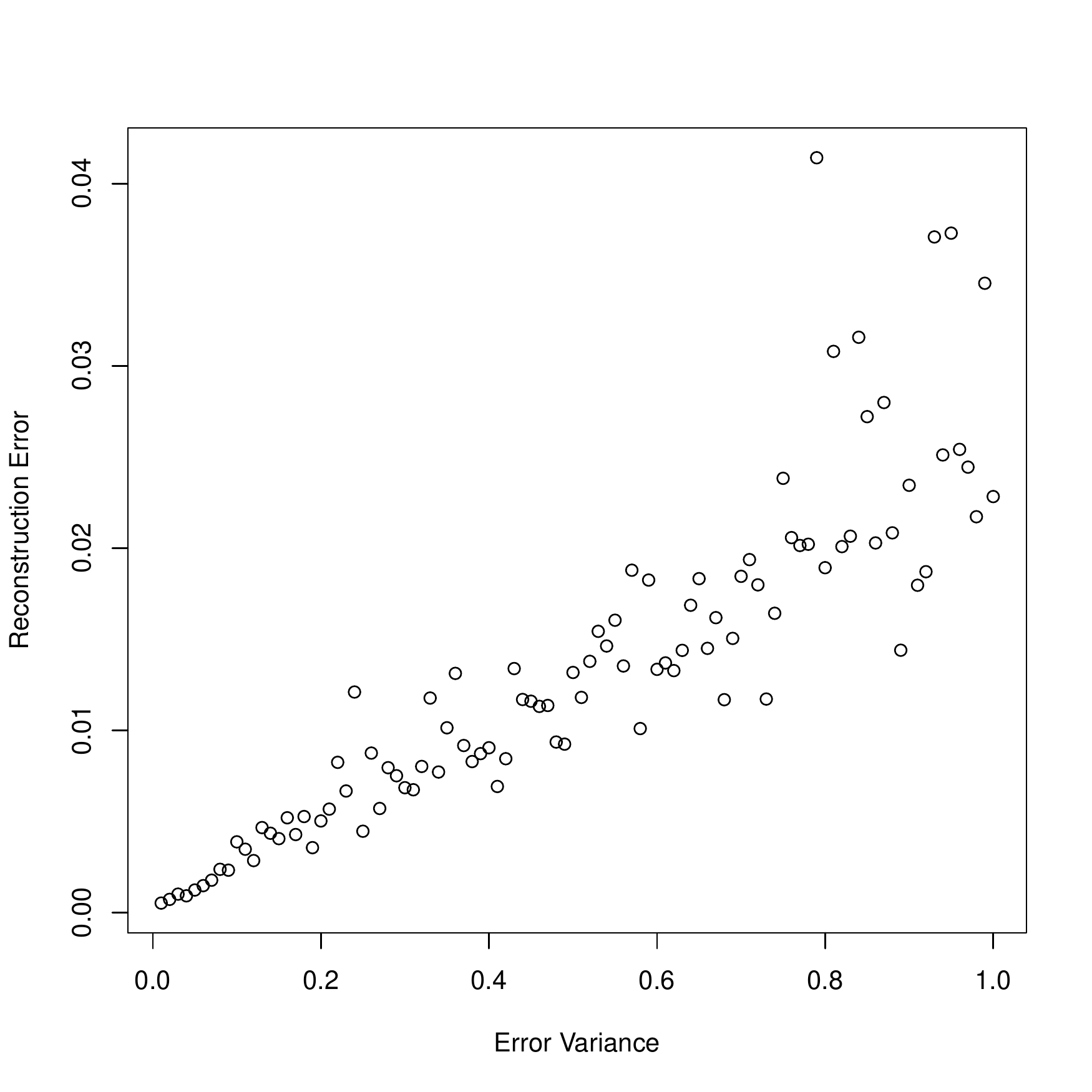}
   \caption{}
   \end{subfigure}
   \begin{subfigure}[b]{0.49\textwidth}
   \includegraphics[scale=0.4]{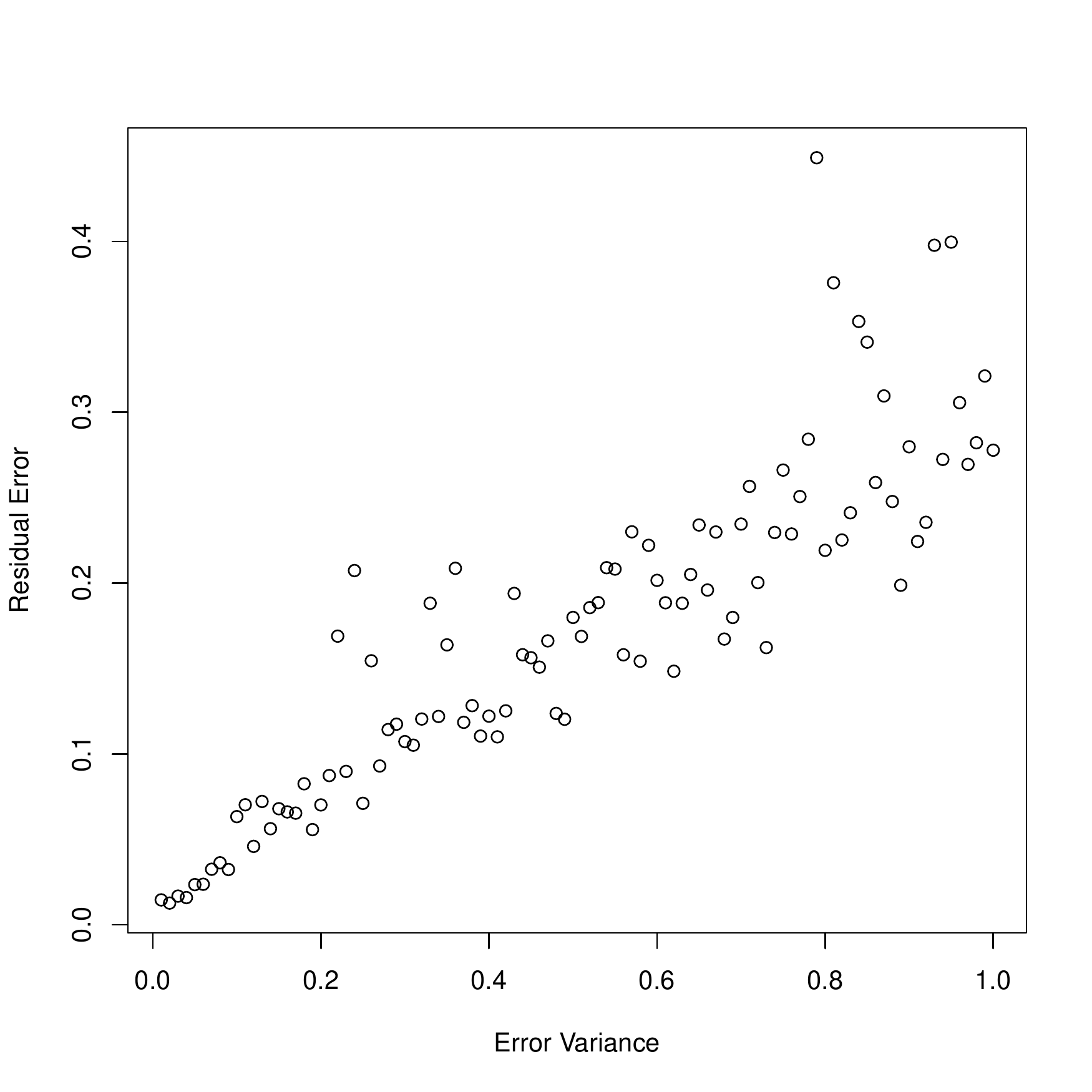}
   \caption{}
   \end{subfigure}
   \caption{The relationship between the error variance of a data set and the reconstruction error (a) or residual error (b) in the LMF decomposition.}
   \label{fig:evarSim}
\end{figure}

\section{Joint and Individual Decomposition}
\label{sec3}
\subsection{Model}

We extend the joint factorization approach of Section~\ref{sec2} to also allow for structured low rank variation that is individual to each data matrix, as follows:
$$ X = US_xV^T + U_{ix}S_{ix}V_{ix}^T + E_x $$
$$ Y = U_yS_yV^T + U_{iy}S_{iy}V_{iy}^T + E_y $$
$$ Z = US_zV_z^T + U_{iz}S_{iz}V_{iz}^T + E_z $$
where the individual components are given by\begin{itemize}
\item $U_{ij}$ for $j=\{x,y,z\}$, matrices representing the row structure unique to each matrix
\item $S_{ij}$ for $j=\{x,y,z\}$, matrices representing the scaling for the individual structure of each matrix, and 
\item $V_{ij}$ for $j=\{x,y,z)\}$, matrices representing the column structure unique to each matrix.
\end{itemize}   We define this joint and individual factorization as \emph{LMF-JIVE}.  In practice, we do not estimate the scaling matrices $S_{ix}$, $S_{iy}$, and $S_{iz}$, and instead allow them to be subsumed into the row and column structure matrices.  

\subsection{Estimation}

To estimate the joint and individual structure, we extend the alternating least squares algorithm from Section \ref{estj}.  Define the low rank approximations for individual structure 
\begin{align}
\label{facEq2}
\begin{split}
 A_x &= U_{ix}S_{ix}V_{ix}^T \\
 A_y &= U_{iy}S_{iy}V_{iy}^T \\
 A_z &= U_{iz}S_{iz}V_{iz}^T.
 \end{split}
\end{align}
  Let $r$ be the rank of joint structure as defined in Section~\ref{sec2}, and let $r_x$, $r_y$, and $r_z$ be the ranks of the individual structure  $A_x$, $A_y$, and $A_z$, respectively.  The estimation algorithm proceeds by iteratively updating the components of $\{J_x,J_y,J_z\}$ and $\{A_x,A_y,A_z\}$ to minimize the total sum of squared residuals until convergence. Thus, to estimate joint structure we define the partial residuals $X^{J}=X - A_x$, $Y^{J}=Y - A_y$, and $Z^{J}= Z - A_z$.  We similarly define $X^{I}=X - J_x$, $Y^{I}=Y - J_y$, and $Z^{I}= Z - J_z$.   In practice we    center and scale the three data sets, as for the joint LMF model. We also initialize $\tilde{V}$, $S_x$, and $U_y$ as in the joint LMF model (see Section \ref{estj}).   For LMF-JIVE, we must also initialize $A_x$, $A_y$, and $A_z$ to matrices of zeros.  We then repeat the following steps until convergence (when the total sum of squares for the joint and individual estimates between the current iteration and the previous iteration is less than a chosen threshold) or until we reach the maximum number of iterations:
\begin{enumerate}
\item Set $X^{J}=X - A_x$, $Y^{J}=Y - A_y$, and $Z^{J}= Z - A_z$.
\item Define concatenations $\tilde{Y} = [(X^{J})^T \,\,\, (Y^{J})^T]^T$ and $\tilde{Z} = [X^{J} \,\,\, Z^{J}]$.
\item Update $U$ via ordinary least squares: $ U = (\tilde{V}^T \tilde{V})^{-1} \tilde{V}^T \tilde{Z} $
\item Scale $U$ by dividing each column by its Frobenius norm.
\item Update $ \tilde{U} = [US_x \,\,\, U_y] $
\item Update $V$ via ordinary least squares $ V = (\tilde{U}^T \tilde{U})^{-1} \tilde{U}^T \tilde{Y} $
\item Update $V_z$ via ordinary least squares: $ V_z = (U^T U)^{-1} U^T Z^{J} $
\item Scale $V$ by dividing each column by its Frobenius norm.
\item Update $U_y$ via ordinary least squares: $ U_y = (V^T V)^{-1} V^T Y^{J} $
\item Update $S_x$ via least squares; define $W: np \times r$ such that the $i$'th column of $W$ is the vectorization of the product of the $i$th columns of $U$ and $V$, $W[,i] = \mbox{vec}(U[,i] V[,i]^T)$, then the diagonal entries of $S_x$ are $(W^T W)^{-1}W^T \mbox{vec}(X)$
\item Recalculate $\tilde{U}$ and $\tilde{V}$ with the newly updated $S_x$.
\item Set $X^{I}=X - A_x$, $Y^{I}=Y - A_y$, and $Z^{I}= Z - A_z$.
\item Update $A_x$ via a rank $r_x$ SVD of $X^{I}$, wherein  $U_{ix}$ gives the right singular vectors, $V_{ix}$ gives the left singular vectors, and $S_{ix}$ gives the singular values.
\item Update $A_y$ via a rank $r_y$ SVD of $Y^{I}$, and update  $A_z$ via a rank $r_z$ SVD of $Z^{I}$
\end{enumerate}

After convergence, we suggest applying the following transformation to assure that the joint and individual structures are orthogonal for $Y$ and $Z$: 
\begin{align}
\begin{split}
\label{trEq}
 J_y^{\perp} &= J_y + A_y P_J^R , \,\,
 A_y^{\perp} = A_y - A_y P_J^R \\
 J_z^{\perp} &= J_z + P_J^C A_z , \,\,
 A_z^{\perp} = A_z - P_J^C A_z.
\end{split}
\end{align}
Here, $ P_J^C = UU^T $ is the projection onto the column space of $J_x$, and $ P_J^R = VV^T $ is the projection onto the row space of $J_x$. This transformation makes the joint and individual structures identifiable; this and other properties  are investigated in Sections~\ref{theory} and~\ref{JIVEsim}. The main scientific rationale for the orthogonalizing transformation is that any structure in the individual matrices that is in the column or row space of $J_x$ should reasonably be considered joint structure.  

In the preceding algorithm we begin with the estimation of joint structure. We could alternatively estimate the individual structure first, by  initializing $J_x$, $J_y$, and $J_z$ to matrices of zeros and having steps 12 and 13 precede step 1 in the algorithm.  We also consider this alternative approach in Section~\ref{simji}.

\subsection{Theoretical Results} \label{theory}

\subsubsection{Uniqueness}
\label{unique}

Below we show that the LMF-JIVE decomposition for the structure of matrix $X$ is unique and identifiable. The proof rests on the defining assumption that shared row and column spaces are the same for joint structure, but different for individual structure.  

\begin{theorem}
\label{thm1}
Assume $ X = J_x + A_x $, $ Y = J_y + A_y $, and $ Z = J_z + A_z $ where
$$ row(J_x) = row(J_y) \; , \; row(A_x)\cap row(A_y)=\{\mathbf{0}\} \; , \; row(J_x)\cap row(A_x)=\{\mathbf{0}\} $$ and
$$ col(J_x) = col(J_z) \; , \; col(A_x)\cap col(A_z)=\{\mathbf{0}\} \; , \; col(J_x)\cap col(A_x)=\{\mathbf{0}\}.$$  If also $ X = J^*_x + A^*_x $, $ Y = J^*_y + A^*_y $, and $ Z = J^*_z + A^*_z $ where
$$ row(J^*_x) = row(J^*_y) \;,\; row(A^*_x)\cap row(A^*_y)=\{\mathbf{0}\} \;,\; row(J^*_x)\cap row(A^*_x)=\{\mathbf{0}\}$$ and
$$ col(J^*_x) = col(J^*_z) \;,\; col(A^*_x)\cap col(A^*_z)=\{\mathbf{0}\} \; , \; col(J^*_x)\cap col(A^*_x)=\{\mathbf{0}\}, $$ then $J_x=J^*_x$ and $A_x=A^*_x$.  
\end{theorem}

\begin{proof}
We first show that $J_x$ and $J^*_x$ have the same row and column spaces. By Theorem 1.1 in the supplement of \citet{lock2013joint} there exists a unique orthogonal decomposition of $X$ and $Y$:
 \begin{align*}
  &X = J^\perp_x + A^\perp_x  &J^\perp_x A^{\perp^T}_x=0_{m_1\times m_1}   \\
   &Y = J^\perp_y + A^\perp_y  &J^\perp_y A^{\perp^T}_y=0_{m_2\times m_2}	
 \end{align*}
such that $row(J_x) = row(J^\perp_x)$ and $row(J^*_x) = row(J^\perp_x)$.  Thus $row(J_x) = row(J^\perp_x)$, and by a symmetric argument $col(J_x) = col(J^*_x)$.      

We next show that $A_x$ and $A^*_x$  have the same row and column spaces, by showing that they have the same null spaces.  Define the nullspace of $A_x$, $N(A_x) = \{\mathbf{v} \in  \mathbb{R}^{n_1}: A_x v = \mathbf{0}\}$,  and take $\mathbf{v} \in N(A_x)$.  If $\mathbf{v} \in N(J_x)$, then 
\[X v = J_x v + A_x v = \mathbf{0}+\mathbf{0} = \mathbf{0},\]
and $\mathbf{v} \in N(J^*_x)$ because $row(J_x) = row(J^*_x)$.  So, $\mathbf{v} \in N(A^*_x)$ because
\[\mathbf{0} = J^*_x v + A^*_x v =A^*_x v.\]
If $\mathbf{v} \notin N(J_x)$, then 
\[X v = J_x v \in col(J_x).\]
So, because $col(J_x)=col(J^*_x)$,
\[J^*_x v +A^*_x v \in col(J^*_x).\]
Thus $A^*_x v=\mathbf{0}$, because $col(J_x)\cap col(A_x)=\{\mathbf{0}\}$, and we again conclude $v \in N(A^*_x)$.  It follows that $N(A_x) \subseteq N(A^*_x)$; symmetric arguments show $N(A^*_x) \subseteq N(A_x)$, $N(A^T_x) \subseteq N(A^{*^T}_x)$, and $N(A^{*^T}_x) \subseteq N(A^T_x)$.  Thus $N(A_x) = N(A^*_x)$ and $N(A^T_x) = N(A^{*^T}_x)$, so $row(A_x) = row(A^*_x)$ and $col(A_x) = col(A^*_x)$.  

Define $J_\text{diff} = J_x-J^*_x$ and $A_\text{diff} = A_x-A^*_x$, and consider that 
\[J_\text{diff} + A_\text{diff} = X-X = 0_{m_1 \times n_1}.\]
Note that $row(J_\text{diff}) \subseteq row(J_x)$ because $row(J_x)=row(J^*_x)$, and $row(A_\text{diff}) \subseteq row(A_x)$ because $row(A_x)=row(A^*_x)$, so $row(J_\text{diff})\cap row(A_\text{diff}) = \{\mathbf{0}\}$. For any $\mathbf{v} \in \mathbb{R}^{n_1}$, 
\[J_\text{diff} \, v + A_\text{diff} \, v = \mathbf{0},\] and therefore $J_\text{diff} \, v = \mathbf{0}$ and $A_\text{diff} \, v = \mathbf{0}$.  It follows that  $J_\text{diff} = A_\text{diff} = 0_{m_1 \times n_1}$, and thus $J_x=J^*_x$ and $A_x=A^*_x$.          
\end{proof}

By the inherent identifiability of the decomposition for $X$, it follows that the entire LMF-JIVE decomposition is identifiable under the orthogonal transformations of  structured variability in $Y$ and $Z$ in Equation~(\ref{trEq}).  The proof is given in Corollary~\ref{cor1}.

\begin{corollary}
\label{cor1}
Assume $\{J_x,A_x,J_y,A_y,J_z,A_z\}$ and $\{J_x^*,A_x^*,J_y^*,A_y^*,J_z^*,A_z^*\}$ satisfy the conditions of Theorem~\ref{thm1}.  Assume, in addition, that the joint and individual structures for $Y$ and $Z$ are orthogonal for both decompositions:
\[J_y A^T_y=0_{m_2\times m_2}, \; J_z^T A_z=0_{n_2\times n_2}, \; \; J^*_y A^{*^T}_y = 0_{m_2\times m_2},\; J^{*^T}_z A^*_z = 0_{n_2\times n_2}. \]
Then $\{J_x,A_x,J_y,A_y,J_z,A_z\} = \{J_x^*,A_x^*,J_y^*,A_y^*,J_z^*,A_z^*\}$.
\end{corollary}

\begin{proof}
By Theorem~\ref{thm1}, $J_x=J^*_x$ and $A_x=A^*_x$.  Define $P_J^R$ as the projection onto the row space of $X$.  Then, \begin{align*}
J_y+A_y&=J_y^* + A_y^* \\
\rightarrow (J_y+A_y)P_J^R &=(J_y^* + A_y^*) P_J^R \\
\rightarrow J_y P_J^R + 0 &= J_y^* P_J^R + 0 \\
\rightarrow J_y = J_y^*. 
\end{align*}
Thus, $A_y=A_y^*$, and an analogous argument shows that $J_z = J_z^*$ and $A_z = A_z^*$.  
\end{proof}

\subsubsection{Orthogonality}

In the classical JIVE algorithm for vertical integration, the joint and individual structures were restricted to be orthogonal for identifiability and a parsimonious decomposition. For LMF-JIVE we restrict the joint and individual structures for $Y$ and $Z$ to be orthogonal,  and thus but not for $X$, which shares a joint row space and column space.  As shown in Section~\ref{unique}, orthogonality between joint and individual structure in $X$, in either the row space or the column space,  is not needed for identifiability of the decomposition.    Thus, variation in $Y$ and $Z$ is parsimoniously decomposed into joint, individual and residual variability, e.g., $||Y||_F^2 = ||J_y||_F^2 + ||A_y||_F^2+||E_y||_F^2$; however, this is not the case for $X$ in general: $||X||_F^2 < ||J_x||_F^2 + ||A_x||_F^2+||E_x||_F^2$.  Interestingly, no such decomposition exists or post-hoc transformation of the converged result exists for $X$ that would yield an orthogonal and parsimonious decomposition.  We show this in Theorem \ref{thm:orth} with a proof by contradiction.

\begin{theorem}
\label{thm:orth}
 
Assume we have $ X = J_x + A_x $, $ Y = J_y + A_y $, and $ Z = J_z + A_z $ where
$$ row(J_x) = row(J_y) $$
$$ col(J_x) = col(J_z) $$
$$ rank(J_x) = rank(J_y) = rank(J_z) = r $$

If we also have $J^*_x$, $J^*_y$, $J^*_z$, $A^*_x$, $A^*_y$, and $A^*_z$ such that
$$ row(J^*_x) = row(J^*_y) $$
$$ col(J^*_x) = col(J^*_z) $$
$$ rank(J^*_x) = rank(J^*_y) = rank(J^*_z) = r $$
and all of the joint and individual matrices are orthogonal:
$$ row(J^*_x) \perp row(A^*_x) $$
$$ col(J^*_x) \perp col(A^*_x) $$
$$ row(J^*_y) \perp row(A^*_y) $$
$$ col(J^*_z) \perp col(A^*_z) $$

Then $J^*_x + A^*_x \ne J_x + A_x = X $.
\end{theorem}

\begin{proof}

We claim that there exist $ X = J^*_x + A^*_x $, $ Y = J^*_y + A^*_y $, and $ X = J^*_y + A^*_y $ that meet the following conditions:
$$ row(J^*_x) = row(J^*_y) $$
$$ col(J^*_x) = col(J^*_z) $$
$$ rank(J^*_x) = rank(J^*_y) = rank(J^*_z) = r $$
We also claim that the joint and individual estimates $J^*$ and $A^*$ are orthogonal:
$$ row(J^*_x) \perp row(A^*_x) $$
$$ col(J^*_x) \perp col(A^*_x) $$
$$ row(J^*_y) \perp row(A^*_y) $$
$$ col(J^*_z) \perp col(A^*_z) $$

Using a result from \citet{lock2013joint}, the only estimate that satisfies $ row(J^*_x) = row(J^*_y) $ and $ row(J^*_x) \perp row(A^*_x) $ under these conditions is $ J^*_x = J_x + A_x P_J^R $, where $P_J^R$ is the orthogonal projection onto the row space of $J_x$.  Similarly, the conditions $ col(J^*_x) = col(J^*_z) $ and $ col(J^*_x) \perp col(A^*_x) $ imply that the estimate is $ J^*_x = J_x + P_J^C A_x $.  However, this implies that $ A_x P_J^R = P_J^C A_x $, which is not necessarily true in general.  Therefore, such an estimate does not exist.  
\end{proof}

\subsection{Simulation Study} \label{JIVEsim}
\label{simji}

Here we present a simulation study to test the ability of the LMF-JIVE algorithm to recover the true structure of the data.  Our simulation design is analogous to that in Section \ref{simst1}, but allows for simulated data sets to have both joint and individual structure.  There were $300$ simulated data sets, with three different variability settings representing equal joint and individual variability, higher joint variability, and higher individual variability.  For the first $100$ simulations, all elements of the joint structure components, $U$, $S_x$, $V$, $U_y$, and $V_z$, and the individual structure components, $U_{ix}$, $V_{ix}$, $U_{iy}$, $V_{iy}$, $U_{iz}$, and $V_{iz}$, were simulated from a Normal$(0,1)$ distribution.  For the next $100$ simulations, the joint components were simulated from a Normal$(0,9)$ distribution, while the individual structure components were simulated from a Normal$(0,1)$ distribution.  In the last $100$ simulations, the joint components were simulated from a Normal$(0,1)$ distribution, and individual components were simulated from a Normal$(0,9)$ distribution.  Error matrices were simulated from a Normal$(0,1)$ distribution.   As with the previous simulation, the simulated data sets had $50$ rows and $50$ columns, and the rank of the joint structure was $2$.  The rank of the individual structure was also $2$ for all three matrices.  We chose rank 2 for all structures in the simulation for simplicity, but the algorithm is not limited to situations of equal rank and rank selection is explored in Section~\ref{sec5}.  

For these simulations, we applied three different versions of the LMF or LMF-JIVE algorithms.  First, we used the LMF algorithm with only joint structure (JO).  We also applied LMF-JIVE, initializing by estimating (1) joint structure first (JF) or individual structure first (IF). We evaluated these algorithms by the same metrics as in the joint structure only simulation: reconstruction error and residual error,  as follows:

$$ E_{rec} = \frac{\sum_{i=x,y,z} (\| J_i - J_i^{true} \| ^2_F + \| A_i - A_i^{true} \| ^2_F)}{\| J_x^{true} \| ^2_F + \| J_y^{true} \| ^2_F + \| J_z^{true} \| ^2_F + \| A_x^{true} \| ^2_F + \| A_y^{true} \| ^2_F + \| A_z^{true} \| ^2_F} $$
$$ E_{res} = \frac{\| J_x + A_x - X \| ^2_F + \| J_y + A_y - Y \| ^2_F + \| J_z + A_z - Z \| ^2_F}{\| X \| ^2_F + \| Y \| ^2_F + \| Z \| ^2_F} .$$

The results are given in Table \ref{tab:sim2}.  In general either of the LMF-JIVE settings performed better than the LMF-only settings, demonstrating the value of distinguishing joint and individual structure.  However, the relative performance of joint-first or individual-first estimation for LMF-JIVE depended on the context, with individual-first estimation performing better in scenarios with higher individual signal and joint-first estimation performing better in scenarios with higher joint signal.  This demonstrates that the algorithm does not always converge to a global least-squares solution.  Thus, the results can be explained because the initial iteration will capture more variability with whichever component is estimated first.  For this reason, we recommend using whichever order estimates the structure with the largest variance first or using both models and comparing them in terms of converged SSE.  Between the joint and individual models, the mean residual error is comparable but those the better performing approach in terms of signal recovery  tends to have a lower residual error.  This is what we should expect, because the algorithm is supposed to be minimizing over the residual error.  This is smaller than the residual error for the joint only model, which is also expected, since we simulated data with both joint and individual structure.  



 \begin{table}
  \centering
    \caption{Reconstruction error and residual error for 3 different variations of the LMF algorithm under equal joint and individual variance, higher joint variance, and higher individual variance.  Key: JO = joint only algorithm; JF = joint structure estimated first; IF = individual structure estimated first.}
  \label{tab:sim2}
 \resizebox{\textwidth}{!}{  
  \begin{tabular}{ | c | c c c | }
    \hline
    $\mathbf{E_{rec}}$ \textbf{Mean (St. Dev.)} & JO & JF & IF \\
    \hline
    Equal Variance & 0.7246 (0.1701) & 0.1846 (0.2003) & 0.7079 (0.7783) \\
    Higher Joint & 0.0098 (0.0019) & 0.0017 (0.0008) & 1.997 (0.0085) \\
    Higher Ind & 1.643 (0.0272) & 0.9870 (0.2284) & 0.0505 (0.0259) \\
    \hline
        $\mathbf{E_{res}}$ \textbf{Mean (St. Dev.)} & JO & JF & IF \\
    \hline
    Equal Variance & 0.4950 (0.0613) & 0.1646 (0.0205) & 0.1787 (0.0201) \\
    Higher Joint & 0.0126 (0.0021) & 0.0038 (0.0005) & 0.0041 (0.0007) \\
    Higher Ind & 0.3591 (0.0269) & 0.0096 (0.0018) & 0.0057 (0.0007) \\
    \hline
  \end{tabular}}

\end{table}

\section{Imputation} \label{imputation}
\label{sec4}

In this section, we use the LMF and LMF-JIVE frameworks to introduce an imputation method for linked data with various forms of missingness, including missing rows and columns.  Our approach extends similar methods that have been created for data imputation with a single matrix using an SVD.  We extend an EM algorithm to iteratively impute the missing values using the SVD and compute the SVD given the imputed values (see \citep{kurucz2007methods,fuentes2006using}.  Thus, our method proceeds by iteratively updating missing values with successive applications of LMF-JIVE.  Here we focus on missing data in $X$ because of our interest in the imputation of cytotoxicity data; however, it is straightforward to extend the approach to impute missing data in $Y$ or $Z$ as well.

\subsection{Algorithm}
\label{imputeAlg}
The imputation algorithm begins by initiating all of the missing values in the data.  For single missing values, the entries are set to the mean of the row and column means for that position.  If a full row is missing, each entry is set to the column mean for each column.  Similarly if a full column is missing, each entry is set to its respective row mean.  Finally, for the entries where both the row and the column are entirely missing, the entries are set to the full matrix mean.  For a matrix X:
$$ \hat{X}_{ij} = \bar{X} \text{ if row i and column j are both missing }  $$
$$ \hat{X}_{ij} = \bar{X}_{i.} \text{ if row i is non-missing but column j is missing } $$
$$ \hat{X}_{ij} = \bar{X}_{.j} \text{ if row i is missing but column j is non-missing } $$
$$ \hat{X}_{ij} = (\bar{X}_{i.} + \bar{X}_{.j}) / 2 \text{ if row i and column j are both non-missing } $$

After initializing the matrix, we do the following:
\begin{enumerate}
\item Compute the LMF-JIVE decomposition using the imputed matrix $\hat{X}$.  Let $J_x$ be the joint component estimate for $X$ and $A_x$ be the individual component estimate for X.  
\item Impute the missing values using the decomposition in step 1.  For missing entries in $X$, set $ X_{ij} = \hat{X}_{ij} $.
\item Repeat steps 1 and 2 until the algorithm converges.  We used the squared Frobenius norm of the difference between the current and previous estimates of X as our convergence criterion, with a threshold of 0.0001.  For the t$^{th}$ iteration, we stop if $\| \hat{X}^{(t)} - \hat{X}^{(t-1)} \| ^2_F < 0.0001$.  
\end{enumerate}

This imputation strategy can be considered an EM algorithm, under a normal likelihood model.  To formalize this, let $\mu_x = USV + U_x S_x V_x$, $\mu_y = USV + U_y S_y V_y$, and $\mu_z = USV + U_z S_z V_z$ give the mean for each entry in a random matrix.  Assume that the residuals from this model are independent and normally distributed with means $\mu_x$, $\mu_y$, and $\mu_z$ for $X$, $Y$, and $Z$, respectively, and variance $\sigma ^2$.  The log likelihood for this model is given below.  Note that the values in $X$, $Y$, and $Z$ are conditionally independent given the parameter space \[\{U, S, V, U_x, S_x, V_x, U_y, S_y, V_y, U_z, S_z, V_z\},\] so the likelihood can be written as a product of independent normal likelihoods.  It is easy to see that this likelihood is maximized when the total sum of squared residuals is minimized, which is accomplished by the alternating least squares method implemented by the LMF algorithm. Thus, step (1.) in the algorithm above corresponds to an M-step.  Step (2.) corresponds to an E-step, where the expected values for $X$, $Y$, and $Z$ are given by their means $\mu_x$, $\mu_y$, and $\mu_z$.  
$$ log L(U, S, V, U_x, S_x, V_x, U_y, S_y, V_y, U_z, S_z, V_z; X, Y, Z) $$
$$ \propto \mbox{log} \prod_{i, j} [ e^{-\frac{1}{2\sigma^2}(X_{ij} - [\mu_x]_{ij})^2} ] \prod_{i, j} [ e^{-\frac{1}{2\sigma^2}(Y_{ij} - [\mu_y]_{ij})^2} ] \prod_{i, j} [ e^{-\frac{1}{2\sigma^2}(Z_{ij} - [\mu_z]_{ij})^2} ] $$
$$ = -\frac{1}{2\sigma^2} [ \sum_{i, j} (X_{ij} - [\mu_x]_{ij})^2 + \sum_{i, j} (Y_{ij} - [\mu_y]_{ij})^2 + \sum_{i, j} (Z_{ij} - [\mu_z]_{ij})^2  ] $$
In practice the residual variance $\sigma^2$ may not be the same across data matrices.  However, when $X$, $Y$ and $Z$ are normalized to have the same Frobenius norm, the M-step (1.) can be considered to optimize a weighted log-likelihood in which the likelihood for each data matrix are scaled to contribute equally.   

\subsection{Simulation}

We generated $100$ data sets using a simulation scheme analogous to that in Section \ref{simji} under each of 6 settings (2 different matrix dimensions for $Y$ and $Z$ and 3 different noise variances).  In each setting, the $X$ matrix dimensions were 50 $\times$ 50.  The $Y$ matrix was $m_2 \times 50$, where $m_2$ was 30 or 200.  Similarly, the $Z$ matrix was $50 \times n_2$, where $n_2$ was 30 or 200.  In all simulations, $m_2$ and $n_2$ were equal.  We varied the non-shared dimensions of $Y$ and $Z$ to see the effect of having more data in $Y$ and $Z$ on the imputation accuracy for $X$.  We also varied the variance of the noise matrices among 0.1, 1, and 10.  The ranks ($r$, $r_x$, $r_y$, and $r_z$) were each chosen from a Uniform($\{0,1,2,3,4,5\})$ distribution.  For each simulation we randomly set 3 rows, 3 columns, and up to 50 additional entries to missing.  We compared 3 different methods for imputation: LMF-JIVE, LMF with only joint structure (LMF), and SVD.  We then evaluated the methods using two different error calculations: the sum of squared errors for the imputed $X$ matrix compared to the simulated $X$ matrix and the sum of squared errors for the imputed $X$ matrix compared to the true underlying structure of the $X$ matrix ($ X_{true} = US_xV^T + U_{ix}S_{ix}V_{ix}^T $).  These values were computed separately for the two types of missing entries (those missing an entire row/column and those missing single entries). 
\begin{align} \text{Error}(X) = \| (X_{est} - X)[\text{missing values}] \| ^2_F / \| X[\text{missing values}] \| ^2_F \label{missingError} \end{align}
$$ \text{Error}(X_{true}) = \| (X_{est} - X_{true})[\text{missing values}] \| ^2_F / \| X_{true}[\text{missing values}] \| ^2_F $$

A summary of the results are given in Table \ref{tab:imp}.  Generally, the LMF-JIVE imputations performs better than the SVD imputation, indicating that incorporating information from $Y$ and $Z$ can improve accuracy. This is especially helpful for the full row/column missing values, where the SVD imputation can not do better than mean imputation because there is no information present for estimating those entries with an SVD (so the SVD error is always greater than 1). An exception  is when there is high noise and the dimensions of $Y$ and $Z$ are small, in which case LMF and LMF-JIVE may be overfitting.    LMF with joint structure universally performed worse than LMF-JIVE for this study, demonstrating the benefit of decomposing joint and individual structure.    

 \begin{table}
  \centering
    \caption{Results of the imputation simulations.  $ \text{Error}(X) $ shows how well each method imputed the simulated data set, while $ \text{Error}(X_{true}) $ shows how well each method imputed the true underlying structure of the simulated data set when excluding random noise.  Each value is the mean from 100 simulations.  The oracle is calculated as $ \| E_x[\text{missing values}] \| ^2_F / \| X[\text{missing values}] \| ^2_F $ and represents the best the imputation can do because of the random noise.}
  \label{tab:imp}
 \resizebox{\textwidth}{!}{  
  \begin{tabular}{ | c c | c c c | c c c | c | }
    \hline
    \multicolumn{2}{|c|}{Individual Missing} & \multicolumn{3}{|c|}{$ \text{Error}(X) $} & \multicolumn{3}{|c|}{$ \text{Error}(X_{true}) $} & \\	
    $m_2, n_2$ & Var(Noise) & SVD & LMF & LMF-JIVE & SVD & LMF & LMF-JIVE & Oracle \\
    \hline
    30 & 0.1 & 0.025 & 0.031 & 0.024 & 0.007 & 0.013 & 0.005 & 0.018 \\
    200 & 0.1 & 0.028 & 0.028 & 0.032 & 0.007 & 0.006 & 0.010 & 0.022 \\
    30 & 1 & 0.232 & 0.274 & 0.218 & 0.068 & 0.120 & 0.051 & 0.174 \\
    200 & 1 & 0.222 & 0.215 & 0.198 & 0.067 & 0.058 & 0.039 & 0.165 \\
    30 & 10 & 1.061 & 1.185 & 0.931 & 1.144 & 1.495 & 0.803 & 0.640 \\
    200 & 10 & 1.139 & 0.927 & 0.85 & 1.288 & 0.815 & 0.580 & 0.631 \\
    \hline
    \multicolumn{2}{|c|}{Row/Column Missing} & \multicolumn{3}{|c|}{$ \text{Error}(X) $} & \multicolumn{3}{|c|}{$ \text{Error}(X_{true}) $} & \\	
    $m_2, n_2$ & Var(Noise) & SVD & LMF & LMF-JIVE & SVD & LMF & LMF-JIVE & Oracle \\
    \hline
    30 & 0.1 & 1.022 & 0.968 & 0.572 & 1.022 & 0.965 & 0.563 & 0.019 \\
    200 & 0.1 & 1.023 & 1.698 & 0.594 & 1.024 & 1.731 & 0.584 & 0.023 \\
    30 & 1 & 1.022 & 3.967 & 0.667 & 1.027 & 4.788 & 0.601 & 0.175 \\
    200 & 1 & 1.018 & 7.971 & 0.618 & 1.021 & 9.639 & 0.536 & 0.175 \\
    30 & 10 & 1.016 & 7.652 & 1.124 & 1.050 & 22.804 & 1.328 & 0.650 \\
    200 & 10 & 1.014 & 5.399 & 0.877 & 1.038 & 14.088 & 0.663 & 0.632 \\
    \hline
  \end{tabular}}
\end{table}

\section{Rank Selection Algorithms}
\label{sec5}

Selection of the joint and individual ranks for LMF-JIVE must be considered carefully, to avoid misallocating joint and individual structure.  Several rank selection approaches have been proposed for JIVE and related methods for vertical integration, including permutation testing \citep{lock2013joint}, Bayesian information criterion (BIC) \citep{o2016r,jere2014}, and likelihood cross-validation \citep{li2017incorporating}. Here we propose and implement two approaches to rank selection in the LMF-JIVE context: (i) a permutation testing approach extending that used for JIVE (Section~\ref{permtest}), and (ii) a novel approach based on cross-validated imputation accuracy of missing values (Section~\ref{CVrankselect}).  Our results suggest that the two approaches give similar overall performance regarding rank recovery; approach (i) is well-motivated if a rigorous and conservative statistical approach to the identification of joint structure is desired, whereas  approach (ii) is well-motivated if missing value imputation is the primary task.   

\subsection{Permutation Test}
\label{permtest}

\subsubsection{Algorithm}

We generate a null distribution for joint structure by randomly permuting the rows of $Z$ and the columns of $Y$ to break any associations between the matrices while maintaining the structure within each matrix.  We generate a null distribution for individual structure by an appropriate permutation of the entries within a matrix, under the motivation that individual low-rank approximations should give correlated structure that is not explained by the joint approximation.  We iterate between selecting the joint rank from $\{X^{J}, Y^{J}, Z^{J}\}$, updating the LMF-JIVE decomposition, and selecting the individual ranks from $\{X^{I},  Y^{I}, Z^{I}\}$, until the selected ranks remain unchanged.  

The joint rank is estimated as follows: 
\begin{enumerate}
\item Let $r_{max}$ be the maximum possible (or plausible) joint rank ($r$) for the data.
\item Initialize $ \tilde{X} = X^{J} $, $ \tilde{Y} = Y^{J} $, and $ \tilde{Z} = Z^{J} $.
\item Compute the sum of squared residuals  for a rank $1$ LMF approximation of the data:  \[ \text{SSR} = \| J_x - \tilde{X} \| ^2_F+\| J_y - \tilde{Y} \| ^2_F +\| J_z - \tilde{Z} \| ^2_F\]
\item Subtract the joint structure from step 3 from $X^{J}$, $Y^{J}$, and $Z^{J}$. Set $ \tilde{X} = \tilde{X} - J_x $, $ \tilde{Y} = \tilde{Y} - J_y $, and $ \tilde{Z} = \tilde{Z} - J_z $.
\item Repeat steps 3 and 4 $(r_{max} - 1)$ times. 
\item Permute the columns of $Y$ and the rows of $Z$ by sampling without replacement, yielding $ Y_{perm} $ and $ Z_{perm} $.
\item Set $ \tilde{Y} = Y_{perm} $, and $ \tilde{Z} = Z_{perm} $.
\item Repeat steps 3 and 4 $r_{max}$ times using the simulated data set. 
\item Repeat steps 6 through 8 for 99 more permutations.
\item Select $r$ as the highest rank such that the SSR of the true data is higher than the 95$^{th}$ percentile of the SSRs for the permuted data.  
\end{enumerate}

To estimate the individual ranks of the data sets, we permute all entries of $X^I$, the entries within each row in $Y^I$, and the entries within each column in $Z^I$.  We compute an SVD of rank  $r_{i, max}$ for each of the true data matrices $i={X,Y,Z}$, and also for each of the permuted matrices.   For each matrix, $r$ is then chosen to be the highest rank such that the $r^{th}$ singular value of the true data is higher than the 95$^{th}$ percentile of the $r^{th}$ singular values for the permuted data sets.  

\subsubsection{Simulation} \label{permsim}

We used a simulation to test the permutation rank selection algorithm.  We simulated 100 data sets with randomly chosen joint and individual ranks from independent Uniform$\left(\{0,1,2,3,4,5\}\right)$ distributions.  As in Section \ref{simji}, all elements of the joint structure components, $U$, $S_x$, $V$, $U_y$, and $V_z$, and the individual structure components, $U_{ix}$, $V_{ix}$, $U_{iy}$, $V_{iy}$, $U_{iz}$, and $V_{iz}$, were simulated from a Normal$(0,1)$ distribution. 


The simulation showed that the permutation test tends to underestimate joint rank, and it tends to overestimate the individual rank.  It underestimated joint rank in 76\% of the simulations.  It overestimated the individual ranks for $X$, $Y$, and $Z$ by 51\%, 59\%, and 68\%, respectively.  The individual rank of $X$ was estimated better than the other ranks, with 43\% of simulations getting the correct rank and a mean absolute deviation of 0.85.

\subsection{Cross-validation}
\label{CVrankselect}

\subsubsection{Algorithm}

We consider the following forward selection algorithm based on cross-validation missing value imputation to choose the ranks:
\begin{enumerate}
\item Randomly set a portion (full rows, full columns, and some individual entries) of the matrix entries to missing.  
\item Use the imputation algorithm described in Section \ref{imputation} to estimate the full matrix, $X_{est}$.  
\item Compute the sum of squared error ($ SSE_{0} = \| (X_{est} - X) \| ^2_F $) for the null imputation with ranks ${r, r_x, r_y, r_z} = {0, 0, 0, 0}$.
\item Add 1 to each rank and compute the SSE for each resulting imputed matrix.
\item If no models were better than $ SSE_{0} $, then choose that model's ranks and stop.  
\item Otherwise, choose the model with the lowest SSE and set $ SSE_{0} $ to that model's SSE.  
\item Repeat steps 4 through 6 until adding 1 to any rank does not decrease the SSE of the imputation.  
\end{enumerate}

While we chose forward selection here, stepwise selection and an exhaustive search of all possible combinations are also possible.  We tested the stepwise selection and got almost equivalent results, indicating that the algorithm was seldom taking backwards steps.  A more exhaustive search would have probably yielded more accurate results, as the SSE was higher for our selected ranks than when using the true ranks in many of the cases; however, this would have required running the algorithm for $ 6^4 $ rank combinations, which is computationally infeasible.

\subsubsection{Simulation}
We tested this cross-validation approach using $100$  simulated data sets generated as in Section \ref{permsim}.  In contrast to the permutation test, the cross-validation method tended to overestimate the joint rank and underestimate the individual ranks.  It correctly estimated joint rank in 30\% of simulations and overestimated it in 58\%.  The mean absolute deviation for the estimated rank from the true joint rank was 1.17.  The estimates for the individual structure of X were closer, with 58\% capturing the true rank and an mean absolute deviation of 0.74.  The $Y$ and $Z$ individual ranks were underestimated 61\% and 62\% of the time, respectively.  They had a mean absolute deviation of 1.26 (for $Y$) and 1.21 (for $Z$).  


\section{Application to Toxicology Data}
\label{sec6}

We applied the LMF-JIVE algorithm to the toxicity ($X$), chemical attribute ($Y$), and genotype data ($Z$).  Previous investigations of these data \citep{abdo2015population,eduati2015prediction} and similar data \citep{lock2012quantitative}  have found clear evidence of systematic cell-line variability in toxicity for several chemicals, but the genetic drivers of this variability have not been well characterized. Various analyses presented in \citet{eduati2015prediction} show that molecular chemical attributes are significantly predictive of overall (mean) chemical toxicity.     We are interested in assessing the combined predictive power of chemical attributes and genetics on toxicity using a fully multivariate and unified approach, and we are also interested more generally in exploring patterns of systematic variability within and across these three linked data matrices.    
%

For $751$ cell lines, the cytotoxic effect of each of $105$ different chemicals was quantified by log(EC$_{10}$), where EC$_{10}$ is the dose concentration that results in a 10\% decrease in cell viability; these values make up the toxicity matrix $X: 751 \times 105$.  Approximately $1.3$ million SNPs were available for each of the $751$ cell lines.  SNP values were of the form $z=\{0,1,2\}$, where $z$ gives the number of minor alleles.  We first removed all SNPs with missing values across any of the $751$ cell lines, and removed those SNPS with a minor allele frequency less than or equal to $5/751 \approx 0.007$.    To identify a set of SNPs with potential relevance to cell toxicity, we performed a simple additive linear regression to predict toxicity from SNP for each (SNP, chemical) pair.  Those SNPs with an association p-value less than $10^{-8}$ for any chemical were included, resulting in $441$ SNPs $Z: 751 \times 441$. For each of the $105$ chemicals, data were available for $9272$ quantitative structural attributes defined using the Simplex representation for molecular structure (SIRMS) \citep{kuz2008hierarchical}.  Attributes with value $0$ for at least $100$ chemicals were removed, leaving $2092$ attributes.  These values were log-transformed ($y=log(y+1)$ and centered so that each attribute had mean $0$.  To identify a set of attributes with potential relevance to cell toxicity, we performed a simple additive linear regression to predict toxicity for each (attribute, cell line) pair.  Those attributes with an association p-value less than $10^{-3}$ for any cell line were included, resulting in 105  attributes $Y: 105 \times 105$.  All three matrices ($X$, $Y$, and $Z$) were centered to have overall mean $0$, and scaled to have the same Frobenius norm.         

We conducted a robust 20-fold cross-validation study to assess the accuracy of recovering underlying structure via missing data imputation.  For each fold 5$\%$ of the columns of $X$ (chemicals) were withheld as missing, 5$\%$ of the rows of $X$ were withheld as missing, and 5$\%$ of the  entries of $X$ from the remaining rows and columns are randomly selected to be withheld as missing.  The folds were non-overlapping, so that each entry of $X$ has its column missing exactly once, its row missing exactly once, and its value missing (with most other values in its row and column present) exactly once, over the $20$ folds.  For each fold we impute all missing values as described in Section~\ref{imputeAlg}, using either a joint-only LMF for \{$X$, $Y$, $Z$\}, an SVD approach for $X$ only, or a joint and individual LMF-JIVE approach.          

We select the model ranks via the forward-selection approach of Section \ref{CVrankselect}, using the mean squared error for imputed values (averaged over row-missing, column-missing, and entrywise-missing imputations) as the selection criteria.  This results in a joint factorization with rank $r=4$, an SVD factorization with rank $r_X=5$, and a joint and individual factorization with ranks $r=3, r_X=4, r_Y=2$, and $r_Z=6$.  
	The relative imputation accuracy (Equation (\ref{missingError}) for each method is shown in Table~\ref{tab:toxres}, broken down by accuracy in imputing entry-wise missing values, column-wise missing, row-wise missing, and values that are missing their entire row and column.  The joint and individual factorization approach performed comparatively well for all types of missing data, demonstrating that it is a flexible compromise between the joint only and individual only (SVD) approaches.  SVD imputation performed better than joint imputation for entry-wise missing data, and similarly to joint and individual imputation.  This suggests that the attributes ($Y$) and genetics ($Z$) provide little additional information on the toxicity for a given [cell line, chemical] pair, given other values in $X$.  However, where no toxicity data are available for a given cell line and chemical, imputation accuracy is substantially improved with the joint approaches that use $Y$ and $Z$ (here SVD imputation using $X$ only can perform no better than a relative error of $1$, because there is no data for the given row and column).  
	
 \begin{table}
  \centering
    \caption{Relative error for missing data imputation under different factorization approaches.}
  \label{tab:toxres}
 \resizebox{\textwidth}{!}{  
  \begin{tabular}{ | c | c c c | }
    \hline
    & LMF & SVD & LMF-JIVE \\
    \hline
    Missing chemical and cell line  & $0.878$ & $1.02$ & $0.854$ \\
   Missing chemical & $0.898$ & $1.02$ & $0.875$ \\
   Missing cell line & $0.203$ & $0.208$ & $0.201$ \\
   Missing entry & $0.164$ & $0.112$ & $0.114$\\
    \hline
  \end{tabular}}
\end{table}

Figure~\ref{fig:heatmaps} shows heatmaps of data for each of the linked matrices, as well as heatmaps from their low-rank approximations resulting from the joint and individual factorizations with the selected ranks.  The toxicity data has several chemicals that are universally more toxic across the cell line panel and several chemicals that are universally less toxic; but there are also some chemicals that demonstrate clear and structured heterogeneity across the cell lines.  The chemical attribute heatmap shows distinctive patterns corresponding to more toxic and less toxic chemicals.  Patterns in the SNP data that are associated with toxicity are less visually apparent.  However, a plot of the cell line scores for the first two joint components in Figure~\ref{fig:rcomponents} reveals a prominent racial effect;  cell lines that are derived from individual from native African populations are distinguished from cell lines derived from non-African populations.  The effect of race on toxicity is not strongly detectable when considering each chemical independently;  under independent t-tests for African vs. non-African populations, the smallest FDR-adjusted \citep{benjamini1995} p-value is $0.05$.  However, a permutation-based test using Distance Weighted Discrimination \citep{wei2016direction} for an overall difference between African and non-African populations across the $105$ chemicals is highly significant (p-value$<0.001$), reinforcing the finding that toxicity profiles can differ by race.    

\begin{figure}[!h]
   \centering
   \includegraphics[width=1\textwidth,trim={2cm 6cm 2cm 4cm},clip]{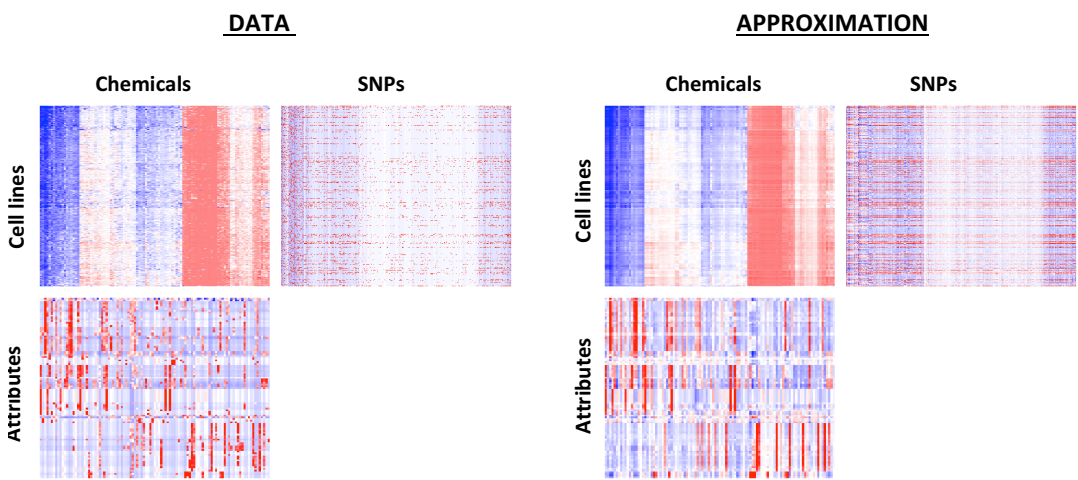}
   \caption{Heatmaps of the toxicity, attribute, and genotype data matrices (left) and their low-rank approximations (right). }
   \label{fig:heatmaps}
\end{figure}

\begin{figure}[!h]
   \centering
   \includegraphics[width=.6\textwidth,trim={0cm 0cm 2cm 2cm},clip]{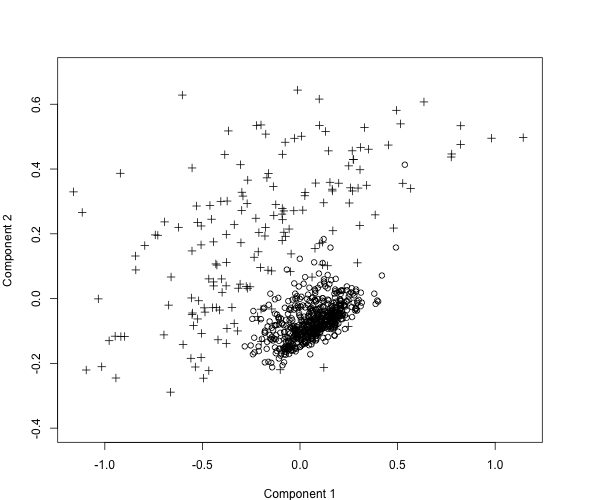}
   \caption{First two cell-line components of joint structure; '+' denotes a cell line derived from a native African population, 'O' denotes a cell line derived from a non-African population.}
   \label{fig:rcomponents}
\end{figure}

\section{Discussion}
\label{sec7}

With a dramatically increasing number of modalities for collecting high-content and multi-faceted data efficiently, large multi-source and interlinked data sets are becoming increasingly common.  Methods that can appropriately address these data without simplifying their structure are needed.  Although our development for this article was motivated by a cytotoxicity study, the methodology is relevant to potential applications in a wide range of other fields. Most notably, new methods for simultaneous horizontal and vertical integration are needed for integrating molecular ``omics" data across multiple disparate sample sets, e.g., the integration of multi-omics molecular data across multiple disparate types of cancer (pan-omics pan-cancer).  LMF improves on previous methods for horizontal or vertical integration only, to allow for bi-dimensional integration.    The LMF and LMF-JIVE decompositions can be used for exploratory analyses, missing data imputation, dimension reduction, and for creating models using the joint or individual components.  

Although we have proposed two rank selection methods for LMF in this paper, both of them tend to be overly conservative.  Other approaches may be more accurate; for example, some form of efficient search may perform better than forward selection for the cross-validation rank selection approach.  It may also be possible to create a model-based approach, such as the BIC approach for JIVE described in \citet{jere2014} and \citet{o2016r}.

As shown in Section \ref{simji}, the LMF-JIVE algorithm does not always reach a global optimum.  Therefore, it may be beneficial to try different starting values for the $U$ or $V$ matrix estimates.  Another option is to try both estimating joint structure first and estimating individual structure first, as in Section~\ref{simji}.  

The LMF algorithm is estimated via a squared residual loss function, which is best motivated under the assumption that the data are normally distributed.  While this is a reasonable assumption for many applications, it may not be appropriate for some contexts.  For example, for categorical or count data, a model-based approach with alternative likelihoods for each dataset may be more appropriate.  This idea is explored for the context of vertical integration in \citet{li2017general}. Adjusting the LMF algorithm to accommodate other likelihoods is an area for future research.

Another exiting direction of future work for the integration of linked data are extensions to higher-order tensors (i.e., multi-way arrays).  Current multi-source decomposition methods, including LMF, are limited to two-dimensional tensors (matrices).  However, multi-source data sets may involve data with more than two dimensions.  For example, for the cytotoxicity data analyzed here, the toxicity matrix is a summary of toxicity data for multiple concentrations of each chemical.  However, we could avoid this summary and work with the individual concentrations, which would give us a 3-dimensional tensor for toxicity (cell lines $\times$ chemicals $\times$ concentrations).  Then we would have a multi-way, multi-source problem that we could solve using LMF if we had a method for handling joint and individual tensor decompositions.  Additionally, a multi-source data decomposition for higher-order tensors could have potential applications in MRI studies and personalized medicine.  LMF provides an important first step toward more general higher-order data integration, because it allows multiple dimensions to be shared, which is  is likely to be encountered for multi-source tensor data.  

\section{Availability}
\label{sec8}

R code to perform the LMF and LMF-JIVE algorithms, including code for rank selection and missing value imputation, are available online at \url{https://github.com/lockEF/LMF}.  

\newpage

\bibliographystyle{apa}
\bibliography{library2}
\end{document}